\documentclass[preprint,12pt]{elsarticle}
\usepackage{amsmath}
\usepackage{ntheorem}
\usepackage{amssymb}
\usepackage{algorithm}
\usepackage{graphicx}
\usepackage{color}
\usepackage{array}
\usepackage{times}
\usepackage{ifpdf}
\usepackage{tikz}
\usetikzlibrary{shapes,snakes}
\def\polylog{\operatorname{polylog}}
\ifpdf \pdfpageheight11in \pdfpagewidth8.5in \fi

\journal{Theoretical Computer Science}
\newtheorem{theorem}{Theorem}
\newtheorem{lemma}[theorem]{Lemma}
\newtheorem{corollary}[theorem]{Corollary}
\newtheorem{problem}{Problem}
\newenvironment{proof}{\noindent {\sc Proof:}}{$\Box$ \medskip}

\include{amssymb}
\newcommand{\etal}{{\em et al.}}
\newcommand{\eps}{\varepsilon}
\newcommand{\dist}{\ensuremath \mathsf{Dist}}

\newcommand{\estimate}{\ensuremath \mathsf{Estimate}}

\newcommand{\disj}{\ensuremath \mathsf{Disj}}
\newcommand{\maxdeg}{\Delta}
\newcommand{\maxedge}{J}
\newcommand{\maxv}{K}
\newcommand{\para}[1]{\medskip \noindent {\bf #1}}
\newcommand{\E}{\mathsf{E}}
\newcommand{\Var}{\mathsf{Var}}

\title{A Second Look at Counting Triangles in Graph Streams \\
(Revised)}
\author[gc]{Graham Cormode\corref{Corresponding author}\fnref{gc},   Hossein Jowhari\fnref{hj}}
\address[gc]{{\tt G.Cormode@warwick.ac.uk}, Corresponding author}
\address[hj]{{\tt hjowhari@sfu.ca}}

\begin{document}

\begin{abstract}
In this paper we present improved results on the problem of counting triangles in edge streamed graphs. 
For graphs with $m$ edges and at least $T$ triangles, we show that an extra look over
the stream yields a two-pass streaming algorithm that uses
$O(\frac{m}{\eps^{2.5}\sqrt{T}}\polylog(m))$ space and
outputs a $(1+\eps)$ approximation of the number of triangles in the graph. This
 improves upon the two-pass streaming tester of Braverman, Ostrovsky and Vilenchik, ICALP 2013, which
 distinguishes between triangle-free graphs and graphs with at least $T$ triangle 
 using $O(\frac{m}{T^{1/3}})$ space. Also, in terms of dependence on $T$, we show that
 more passes would not lead to a better space bound. In other words, we prove there
 is no constant pass streaming algorithm that distinguishes between triangle-free
 graphs from graphs with at least $T$ triangles using $O(\frac{m}{T^{1/2+\rho}})$ 
 space for any constant $\rho \ge 0$. 
\end{abstract}
\maketitle
\section{Introduction}
\label{sec:intro}
Many applications produce output in form of graphs, defined an edge at
a time. 
These include social networks that produce edges corresponding to new
friendships or other connections between entities in the network; 
communication networks, where each edge represents a communication
(phone call, email, text message) between a pair of participants; 
and web graphs, where each edge represents a link between pages. 
Over such graphs, we wish to answer questions about the induced graph,
relating to the structure and properties. 

One of the most basic structures that can be present in a graph is a
triangle: an embedded clique on three nodes. 
Questions around counting the number of triangles in a graph have been
widely studied, due to the inherent interest in the problem, and
because it is a necessary stepping stone to answering questions around
more complex structures in graphs. 
Triangles are of interest within social networks, as they indicate
common friendships: two friends of an individual are themselves
friends. 
Counting the number of friendships within a graph is therefore a
measure of the closeness of friendship activities. 
Another use of the number of triangles is as a parameter for evaluation of large
graph models~\cite{Leskovec:Backstrom:Kumar:Tomkins:08}.

For these reasons, and for the fundamental nature of the problem,
there have been numerous studies of the problem of counting or
enumerating triangles
in various models of data access: 
  external memory~\cite{Liu:Wang:Zou:Wang:10,Hu:Tao:Chung:13}; 
map-reduce~\cite{Suri:Vassilvitskii:11,Pagh:Tsourakakis:12,Tsourakakis:Kang:Miller:Faloutsos:09}; 
and 
  RAM model~\cite{Schank:Wagner:05,Tsourakakis:08}.
Indeed, it seems that triangle counting and enumeration is becoming a
{\em de facto} benchmark for testing ``big data'' systems and their
ability to process complex queries. 
The reason is that the problem captures an essentially hard problem
within big data: accurately measuring the degree of correlation. 
In this paper, we study the problem of triangle counting over
(massive) streams of edges. 
In this case, lower bounds from communication complexity can be
applied to show that exactly counting the number of triangles
essentially requires storing the full input, so instead we look for
methods which can approximate the number of triangles. 
In this direction, there has been series of works that have attempted to 
capture the right space complexity for algorithms that approximate the number
of triangles. 
However most of these works
have focused on one pass algorithms and thus, due to the hard nature of the problem, their space 
 bounds have become complicated, suffering from dependencies on
 multiple graph parameters such as maximum degree, number of 
 paths of length 2, number of cycles of length 4, etc. 
 
 In a recent work
 by Braverman \etal~\cite{Braverman:Ostrovsky:Vilenchik:13}, it has been
 shown that at the expense of an extra pass over stream, a straightforward sampling strategy
 gives a sublinear bound that depends only on $m$ (number of edges) and $T$ (a lower bound
 on the number of triangles\footnote{In this and prior works, some
   assumption on the number of triangles is required.  
This is due in part to the fact that distinguishing triangle-free
graphs from those with one or more triangle requires space
proportional to the number of edges. 
Other works have required even stronger assumptions, such as a bound
on $T_2$, the number of paths of length 2, or the maximum degree of
the graph}). More precisely \cite{Braverman:Ostrovsky:Vilenchik:13}
have shown that one extra pass yields an algorithm  that
distinguishes between triangle-free graphs from graphs with at least $T$
triangles using $O(\frac{m}{T^{1/3}})$ words of space. 
Although their algorithm does not give an estimate of the
number of triangles and more important is not clearly superior to the 
$O(\frac{m\Delta}{T})$ one pass algorithm by \cite{Pagh:Tsourakakis:12,Pavan:Tangwongsan:Tirthapura:Wu:13}
(especially for graphs with small maximum degree $\Delta$), it creates some hope that perhaps with the expense 
of extra passes one could get improved and cleaner space complexities that beat the one pass bound
for a wider range of graphs. 
In particular one might ask is there a $O(\frac{m}{T})$ space multi-pass 
algorithm? 
In this paper, while we refute such a possibility, we show that a more modest bound is possible. 
{Specifically here we show that the sampling strategy of 
\cite{Braverman:Ostrovsky:Vilenchik:13}, namely uniform sampling
of the edges at a rate of $\frac{1}{\sqrt{T}}$ in the first pass
and counting detected triangles in the second
pass gives a $O(1)$ approximation of the number
of triangles. To bring down the approximation precision
 to $1+\eps$, we use a simple summary structure for identifying 
heavy edges (edges shared by many triangles which introduce
large variance in the estimator) in order to deal with them
separately from the rest of the graph. It turns out the right threshold for
{\em heaviness} is $O(\sqrt{t/\eps})$ which can be obtained from the 
two pass constant factor approximation. In order to avoid a third pass,
we run the algorithm in parallel for different guesses of $t$ and at the end
pick the outcome of the guess that matches our constant factor 
approximation of $t$.
We remark that a similar idea has been 
used in the recent work of Eden \etal~\cite{Eden:Levi:Ron:Seshadhri:15} 
for approximately counting triangles in sublinear time.
There, the notion of heaviness is applied to nodes, not edges, and the
model allows query access to node degrees and edge presence.  
In our algorithm, we also utilize the one pass algorithm of Pagh and 
Tsourakakis \cite{Pagh:Tsourakakis:12} (explained below) as
a subroutine. Lastly,     
we observe that this $m/\sqrt{T}$ dependence is attainable in one pass
for a constant factor approximation---under the
stronger assumption of random ordering of edge arrivals.}

Furthermore, via a reduction
to a hard communication complexity problem, we demonstrate that this
bound is optimal in terms of its dependence on $T$. 
In other words there is no
constant pass algorithm that distinguishes between
triangle-free graphs from graphs with at least $T$ triangles using
$O(\frac{m}{T^{1/2+\rho}})$ for any constant $\rho > 0$. 
Our results are summarized in Figure~\ref{fig:results} and compared to
other bounds in terms of the
problem addressed, bound provided, and number of passes.

In line with prior work, we assume a simple graph---that is, each
edge of the graph is presented exactly once in the stream.
Note that our lower bounds immediately hold for the case
when edges are repeated. 

\para{Algorithms for Triangle Counting in Graph Streams.}
The triangle counting problem has attracted particular attention in
the model of graph streams: there is now a substantial body of work in this setting. 
Algorithms are evaluated on the amount of space that they require, the
number of passes over the input stream that they take, and the time
taken to process each update. 
Different variations arise depending on whether deletions of edges are
permitted, or the stream is `insert-only'; and whether arrivals are
ordered in a particular way, so that all edges incident on one node
arrive together, or arrivals are randomly ordered, or adversarially
ordered. 

The work of Jowhari and Ghodsi~\cite{Jowhari:Ghodsi:05}
first studied the most popular of these combinations: insert-only,
adversarial ordering. 
The general approach, common to many streaming algorithms, is to build
a randomized estimator for the desired quantity, and then repeat this
sufficiently many times to provide a guaranteed accuracy. 
Their approach begins by sampling an edge uniformly from the stream of
$m$ arriving edges on $n$ vertices. 
Their estimator then counts the number of triangles incident on a
sampled edge. 
Since the ordering is adversarial, the estimator has to keep track of
all edges incident on the sampled edge, which in the worst case
is bounded by $\maxdeg$, the maximum degree. 
The sampling process is repeated 
$O(\frac{1}{\eps^2} \frac{m\maxdeg}{T})$ times (using the assumed
lower bound on the number of triangles, $T$), leading to a total space
  requirement proportial to 
$O(\frac{1}{\eps^2} \frac{m\maxdeg^2}{T})$ to give an $\eps$
  relative error estimation of $t$, the (actual) number of triangles in the graph.
The parameter $\eps$ ensures that the error in the count is at most 
$\eps t$ (with constant probability, since the algorithm is
randomized). 
The process can be completed with a single pass over the input.
Jowhari and Ghodsi also consider the case where edges may be deleted,
in which case a randomized estimator using ``sketch'' techniques is
introduced, improving over a previous sketch algorithm due to
Bar-Yossef \etal~\cite{Bar-Yossef:Kumar:Sivakumar:02}.

The work of Buriol \etal~\cite{Buriol:Frahling:Leonardi:Marchetti-Spaccamela:Sohler:06}
also adopted a sampling approach, and built a one-pass estimator with smaller
working space. 
An algorithm is proposed which samples uniformly an edge from the
stream, then picks a third node, and scans the remainder of the stream
to see if the triangle on these three nodes is present. 
Recall that $n$ is the number of nodes in the graph, $m$ is number of edges, 
and $T\le t$ is lower bound on the (true) number of triangles. 
To obtain an accurate estimate of the number of triangles in the graph,
this procedure is repeated independently 
$O(\frac{mn}{\eps^2T})$ times to achieve $\epsilon$ relative error. 

Recent work by Pavan \etal~\cite{Pavan:Tangwongsan:Tirthapura:Wu:13}
extends the sampling approach of Buriol \etal: instead of picking a
random node to complete the triangle with a sampled edge, their
estimator samples a second edge that is incident on the first sampled
edge. 
This estimator is repeated $O(\frac{m \maxdeg}{\eps^2 T})$ times, 
where $\maxdeg$ represents the maximum degree of any node. 
That is, this improves the bound of Buriol \etal\ by a factor of 
$n/\maxdeg$. 
In the worst case, $\maxdeg=n$, but in general we expect $\maxdeg$ to
be substantially smaller than $n$.

Braverman \etal~\cite{Braverman:Ostrovsky:Vilenchik:13} take a
different approach to sampling. 
Instead of building a single estimator and repeating, their algorithms
sample a set of edges, and then look for triangles induced by the
sampled edges. 
Specifically, an algorithm which takes two passes over the input
stream distinguishes triangle-free graphs from those with $T$
triangles in space $O(mT^{-1/3})$. 

{For graphs with $W \ge m$ where $W$ is the number of wedges (paths of length 2), 
Jha \etal~\cite{Jha:Seshadhri:Pinar:13} have shown
a single pass $O(\frac{m}{\eps^2\sqrt{T}})$ space algorithm that
returns an additive error estimation of the number of triangles where the 
estimation error is bounded by $\eps W$.}     

Pagh and Tsourakakis~\cite{Pagh:Tsourakakis:12} propose an algorithm
in the MapReduce model of computation, which depends on the maximum
number of triangles on a single edge ($J$).
However, it can naturally be adapted to the streaming setting. 
{As described in Section \ref{sec:upperbounds},
we make use of this algorithm as a subroutine in the design of
our two pass algorithm.
The space used by this algorithms scales as
$O(\frac{mJ}{T} + \frac{m}{\sqrt{T}})$.
}

\para{Lower bounds for triangle counting.}
A lower bound in the streaming model is presented by
Bar-Yossef \etal~\cite{Bar-Yossef:Kumar:Sivakumar:02}.
They argue that there are (dense) families of graphs over $n$ nodes
such that any algorithm that approximates the number of triangles must
use $\Omega(n^2)$ space. 
The construction essentially encodes $\Omega(n^2)$ bits of
information, and uses the presence or absence of a single triangle to
recover a single bit. 
Braverman \etal~\cite{Braverman:Ostrovsky:Vilenchik:13}
show a lower bound of $\Omega(m)$ by demonstrating a family of graphs
with $m$ chosen between $n$ and $n^2$. 
Their construction encodes $m$ bits in a graph, then adds $\tau$ edges
such that there are either $\tau$ triangles or 0 triangles, which reveal
the value of an encoded bit. 

For algorithms which take a constant number of passes over the input stream, 
Jowhari and Ghodsi~\cite{Jowhari:Ghodsi:05} show that still 
$\Omega(n/T)$ space is needed to approximate the number of triangles
up to a constant factor, based on a similar encoding and testing argument. 
Specifically, they create a graph that encodes two binary strings, so
that the resulting graph has $T$ triangles if the strings are
disjoint, and $2T$ if they have an intersection. 
In a similar way, 
Braverman \etal~\cite{Braverman:Ostrovsky:Vilenchik:13}
encode binary strings into a graph, so that it either has no triangles
(disjoint strings) or at least $T$ triangles (intersecting strings). 
This implies that $\Omega(m/T)$ space is required to distinguish the
two cases. 
In both cases, the hardness follows from the communication complexity
of determining the disjointness of binary strings. 

\para{Revision note.}
{
This paper is a revision of an earlier version which claimed the same
main two pass dependence on $T$.
However, the algorithm presented in the earlier version can only
obtain a constant factor approximation.
In this revision, our modified algorithm is based on the same general idea,
that `heavy' edges with many incident triangles are those that prevent
simple sampling-based algorithms from succeeding, and handling such
heavy edges separately can allow accurate algorithms.
Our modified algorithm more directly handles heavy edges, and so can
provide the claimed bounds.}

\section{Preliminaries and Results}
\label{sec:prelims}

\begin{figure}
\centering
\renewcommand{\arraystretch}{1.1}
\begin{tabular}{|c|l|}
\hline
$n$ & number of vertices \\
$m$ & number of edges \\
$t(G)$ & number of triangles in graph $G$\\
$T$ & lower bound on $t(G)$ \\
$\eps$ & relative error \\
$\delta$ & probability of error \\
$\maxdeg$ & maximum degree\\
$t(e)$ & number of triangles that share the edge $e$\\
$\maxedge$ &$ \max_{e\in E} t(e)$\\
$\maxv$ & maximum number of triangles incident on a vertex\\
\hline
$\dist(T)$ & Distinguish graphs with $T$ triangles from
triangle-free graphs\\
$\estimate(T,\eps)$ & $1+\eps$ approximate the number of triangles when there
are at least $T$ \\
$\disj_p^{r}$ & Determine if two length $p$ bitstrings of weight $r$
intersect\\
\hline
\end{tabular}
\caption{Table of notation}
\label{fig:notation}
\end{figure}

In this section, we define additional notation and define the problems
that we study. 

As mentioned above, we use
$t(G)$ to denote the number of triangles in a graph $G = (V,E)$. 
Let $\maxedge(G)$ denote the maximum number of triangles that share an
edge in $G$, and
$\maxv(G)$ the maximum number incident on any vertex. 
We use $t$, $\maxedge$ and $\maxv$ when $G$ is clear from the context. 
 
\para{Problems Studied.}
We define some problems related to counting the number of
triangles in a graph stream. 
These all depend on a parameter $T$ that gives a promise on the number of
triangles in the graph. 

  $\dist(T)$: Given a stream of edges, distinguish
graphs with at least $T$ triangles from triangle-free graphs.


  $\estimate(T,\epsilon)$: Given the edge stream of a graph with at least 
$T$ triangles, output $s$ where  $(1-\epsilon)\cdot t(G)\le  s \le (1+\epsilon) \cdot t(G) $. 

Observe that any algorithm which promises to approximate the number of
triangles for $\epsilon<1$ must at least be able to distinguish the case of
0 triangles or $T$ triangles.  
Consequently, we provide lower bounds for the $\dist(T)$ problem, and
upper bounds for the $\estimate(T,\epsilon)$ problem. 
Our lower bounds rely on the hardness of well-known problems from
communication complexity.  
In particular, we make use of the hardness of $\disj_p^r$: 


\begin{problem}
The  $\disj_{p}^{r}$ problem involves two players, Alice and Bob,
who each have binary vectors of length $p$. 
Each vector has Hamming weight $r$, i.e. $r$ entries set to one. 
The players want to distinguish non-intersecting
inputs from inputs that do intersect. 
\end{problem}

This problem is
``hard'' in the (randomized) communication complexity
setting: it requires a large amount of communication between the
players in order to provide a correct answer with sufficient
probability~\cite{Kushilevitz:Nisan:97}. 
Specifically, 
$\disj_{p}^{r}$ requires $\Omega(r)$ bits of communication for any
$r\le p/2$, over multiple rounds of interaction between Alice and
Bob. 

\begin{figure}
\centering
\renewcommand{\arraystretch}{1.2}
\begin{tabular}{|l|l|l|l|}
\hline 
Problem & Passes & Bound & Reference \\
\hline
$\dist(T)$ & 1 & $\Omega(m)$ & \cite{Braverman:Ostrovsky:Vilenchik:13} \\
$\dist(T)$ & $O(1)$ & $\Omega(m/T)$ & \cite{Braverman:Ostrovsky:Vilenchik:13} \\
$\dist(T)$ & 2 & $O(\frac{m}{T^{1/3}})$ & \cite{Braverman:Ostrovsky:Vilenchik:13} \\
$\estimate(T,\eps)$ & 1 & 
$O(\frac{1}{\eps^2} \frac{m \maxdeg}{T})$ & 
\cite{Pavan:Tangwongsan:Tirthapura:Wu:13} \\
$\estimate(T,\eps)$ & 1 &
$O(\frac{1}{\eps^2}(\frac{m\maxedge}{T} + \frac{m}{\sqrt{T}}))$ &
  \cite{Pagh:Tsourakakis:12} \\
\hline
$\estimate(T,\eps)$ & 2 & $\tilde{O}(\frac{m}{\eps^{2.5}{\sqrt{T}}})$ & Theorem~\ref{thm:sqrtt} \\
$\dist(T)$ & $O(1)$ & $\Omega(\frac{m}{T^{2/3}})$ & Theorem~\ref{thm:lb1} \\
$\dist(T)$ & $O(1)$ & $\Omega(\frac{m}{\sqrt{T}})$ for $m=\Theta(n\sqrt{T})$ & Theorem~\ref{thm:lb2} \\
\hline
\end{tabular}
\caption{Summary of results}
\label{fig:results}
\end{figure}

\para{Our Results.}
We summarize the results for this problem discussed in
Section~\ref{sec:intro}, and include our new results, in
Figure~\ref{fig:results}.
{
We observe that, in terms of dependence on $T$, we achieve tight
bounds for 2 passes: 
Theorem~\ref{thm:sqrtt} shows that we can obtain a dependence on
$T^{-1/2}$, and Theorem~\ref{thm:lb2} shows that no improvement for
constant passes as a function of $T$ can be obtained. 
It is useful to contrast to the results of \cite{Pagh:Tsourakakis:12},
where a one pass algorithm achieves a dependence of $\frac{m}{T^{1/2}}$, but
has an additional term of $\frac{m\maxedge}{T}$. 
This extra term can be large: as big as $m$ in the case that all
triangles are incident on the same edge; here, we show that this term
can be avoided at the cost of an additional pass, in order to identify
edges with more than $\sqrt{T}$ triangles and handle them separately. 
Our results improve over the 2-pass bounds given in~\cite{Braverman:Ostrovsky:Vilenchik:13}.  
Comparing with the additive estimator of \cite{Jha:Seshadhri:Pinar:13}, while our sampling 
strategy is somewhat similar, using an extra pass over the stream we return
a relative error estimation of the number of triangles. 
}

Our analysis assumes familiarity with techniques from randomized algorithms:
first, second and exponential moments methods, in the form of the Markov
inequality, Chebyshev inequality, and Chernoff bounds~\cite{Motwani:Raghavan:95}.


\section{Upper bounds}
\label{sec:upperbounds}
{In this section, we provide an upper bound in the form of
a randomized algorithm which succeeds with constant probability.
We begin by describing a 2-pass algorithm that outputs a constant
factor approximation of $t(G)$ using a lower bound $T$ on $t(G)$ (Section~\ref{sec:const}).
Next we describe our main algorithm that uses the constant factor 
approximation algorithm as a sub-procedure and a summary of the graph
(computed in the first pass) to improve the approximate factor to
$1+\eps$ (Section~\ref{sec:rel}).
We make use of the 
following result by Pagh and Tsourakakis~\cite{Pagh:Tsourakakis:12}.}
 
\begin{lemma}[\cite{Pagh:Tsourakakis:12}]
\label{lem:PT} Given a simple graph $G$ and arbitrary integer $T$, 
there is a one-pass randomized streaming algorithm that outputs 
$t'$ such that $|t'-t(G)|\le \max{\{\eps T,\eps t(G)\}}$. The expected space
usage of the algorithm is  
$\Tilde{O}({1 \over \eps^2}({mJ \over T}+{m \over \sqrt{T}}))$, where
$\maxedge$ denotes the maximum number of triangles incident on a
single edge.   
 \end{lemma}

The algorithm of Lemma~\ref{lem:PT} works by conceptually assigning a
``color'' to each vertex randomly from $C$ colors (this can be
accomplished in the streaming setting with a suitable hash function, for example).
The algorithm then stores each monochromatic edge, i.e. each edge from the input
such that both vertices have the same color. 
Counting the number of triangles in this induced graph, and scaling up
by a factor of $C^2$ gives an estimator for $t$. 
The space used is $O(m/C)$ in expectation. 
Setting $C$ appropriately yields a one-pass algorithm with space
$\tilde{O}(\frac1{\eps^2}({m \over T} \maxedge + {m \over \sqrt T}))$.


\subsection{Constant-factor approximation}
\label{sec:const}
{
The following simple lemma is a key observation in our algorithms. 
Here $E_h$ is the set of all edges $e \in E$ with $t(e) \ge h$.}

\begin{lemma}
\label{lem:heavy}
The number of triangles that contain two or three edges from the
set $E_h$ is less than $({3t\over h})^2$. 
 \end{lemma}
 
 \begin{proof}
{
   From $\sum_{e \in E}t(e) =3t$ it follows that $|E_h| \le {3t\over h}$. Since every 
 two distinct edges belong to at most one triangle, the number of triangles that contain
 two or more edges from $E_h$ is at most ${3t/h \choose 2} < ({3t\over h})^2. $}
 \end{proof}

\begin{algorithm}[t]

Repeat the following  $l \ge 16/\eps$ 
times independently in parallel and output the minimum of the outcomes.\\
{\em Pass 1.} Pick every edge with probability 
$p=O(\frac1{\eps^{4.5}\sqrt{T}})$ (with large
enough constants).\\
{\em Pass 2.} Define $r$ to be the number of triangles that are
observed where two edges were sampled in the first pass, and the
completing edge is seen in the second pass.  
Output  ${r \over {3p^2(1-p)}}$. 
 \caption{The $(3+\eps)$ Algorithm}
 \label{alg:const}
\end{algorithm}

Algorithm~\ref{alg:const} describes our two-pass, $(3 +
\eps)$-factor approximation algorithm.  
Any use of this algorithm will set $\eps$ to be a constant, but for
completeness our analysis makes explicit the dependence on $\eps$.  

\begin{theorem} 
\label{thm:sqrtt}
Algorithm~\ref{alg:const} is a 2-pass randomized streaming algorithm
that 
uses $O( {\frac {m}{\eps^{4.5}\sqrt{T}}})$ space in expectation and 
outputs a $(3+\eps)$ factor approximation of $t$ with constant probability. 
\end{theorem}

\begin{proof} 
Let $\cal{T}$ represent the set of triangles in the graph. 
Consider one instance of the basic estimator, and let
$X$ be the outcome of this instance. 
Let $X_i$ denote the indicator random
variable associated with the $i$th triangle in $\cal{T}$ being detected. 
By simple calculation, we have $\Pr[X_i=1]=3p^2(1-p)$
and $\E(X)=\frac1{3p^2(1-p)} \sum_{i \in \cal{T}}X_i =t$.
Thus, $X$ is an unbiased estimator for $t$; however, $R$, which is the
minimum of $l$ independent repetitions of $X$, is biased. 
By the Markov inequality, $\Pr[X \ge (1+\eps)\E(X)] \le 1/(1+\eps)$.
Therefore, picking $\eps \leq 1$, we can conclude, 
$$\Pr[R \le (1+\eps)t] \ge (1 - \Pr[X \ge (1+\eps)t]^{16/\eps})
\ge 1- \frac12^{16}\ge 1-10^{-4}.$$

However, proving a lower bound on $R$ is more complex, and requires a
more involved analysis. 
First, we notice that most triangles share an edge with a limited number
of triangles.
Let $h=\sqrt{9t \over \eps}$. We call $E_h$ the set of heavy edges.  
Let ${\cal{T}}_h$ denote the set of triangles with two or 
three heavy edges. From Lemma \ref{lem:heavy}, we have that 
 $|{\cal{T}}_h| < \eps t$.
 
  Let $S={\cal{T}}/{\cal{T}}_h$. For each triangle $i \in S$, fix two of its light (non-heavy) edges. 
Let $Y_i$ denote the indicator random variable for the event where
the algorithm picks these two light edges of $i \in S$ in the first pass.
We have $\E(Y_i) = p^2$ and always $Y_i \le X_i$.
Let $Y=\frac1{p^2}\sum_{i \in S} Y_i.$ Assuming $p<1$, by definition we have 
$1/3Y < X$. Therefore a lower bound on $Y$ will give us a lower bound on $X$.
We have $$\E(Y) = |S| \ge (1-\eps)t.$$
Also, \[\Var(Y) = \E(Y^2)-\E^2(Y) \le \frac1{p^2}|S| +
\frac1p|S|\sqrt{t/\eps}.\] 
The first term comes from 
$\sum_{i \in S} \frac1{p^4}\E(Y_i^2)$, and the 
second term arises from pairs of
triangles which share a light edge, of which there are at most
$|S|\sqrt{t/\eps}$ (since the edge is light), and which are both
sampled with probability $p^3$.
Using the Chebyshev inequality and assuming $\eps < \frac12$, we have

\begin{eqnarray*}
 \Pr[Y< (1-\eps)^2t] &\le & \Pr[Y < (1-\eps)|S| ]  \\
   &\le & {\Var(Y) \over {\eps^2|S|^2}}  \\
   &\le & \frac1{\eps^2}\left({1\over p^2|S|} +{\sqrt{t/\eps}\over {p|S|}}\right)  \\
   &< & \frac1{\eps^2}\left({2\over p^2t}
    +{2\over {p\sqrt{\eps t}}}\right).
\end{eqnarray*}

 Since $T \le t$, setting $p>\frac{320}{\eps^{3.5}\sqrt{T}}$,
 allows the above probability to be bounded by $\frac{\eps}{160}$. 
Now the probability that the minimum of $16/\eps$ independent
trials is below the designated threshold is at most
$\frac{\eps}{160}\frac{16}{\eps} =1/10$. 
Therefore with probability at least $1-(1/10^{-4} + 1/10)$ the output of 
the algorithm is within the interval $[1/3(1-2\eps)t,(1+\eps)t]$.
 This proves the statement of our theorem.  
\end{proof}

{
It can be shown the above analysis is tight~\footnote{An earlier
  version of this paper erroneously claimed that this algorithm
  achieved a $1+\eps$ approximation; this example shows that this is
  not the case.}. Consider the
``crown''-like graph where $t$ triangles share an edge shown in Figure~\ref{fig:crown}. If we sample edges at a rate of
$p=O(\frac1{\sqrt{t}})$, the bottom edge is picked with probabilty $p$, an unlikely
event. That means the random variable $r$ will be concentrated around $tp^2$ which divided
by $3p^2(1-p)$ gives roughly $t/3$ as the estimate for number of triangles. On the other
hand, for $t$ disjoint triangles, $r$ is concentrated around $3tp^2(1-p)$ which divided by $3p^2(1-p)$ 
gives the right estimate for the number of triangles.} 

\begin{figure}[t]
\centering
\begin{tikzpicture}

\node (e) at (3,2.7) [] {$t$};
\draw[decorate,decoration={brace},thick] (0,2.3) -- (6,2.3);
        
\node (a) at (2,0) [draw,circle,fill=black,scale=0.7] {};
\node (b) at (4,0) [draw,circle,fill=black,scale=0.7] {};

\node (c) at (0,2) [draw,circle,fill=black,scale=0.7] {};
\node (d) at (1,2) [draw,circle,fill=black,scale=0.7] {};
\node (d) at (2,2) [draw,circle,fill=black,scale=0.7] {};
\node (e) at (3,2) [] {...};
\node (d) at (4,2) [draw,circle,fill=black,scale=0.7] {};
\node (f) at (5,2) [draw,circle,fill=black,scale=0.7] {};
\node (g) at (6,2) [draw,circle,fill=black,scale=0.7] {};

\draw (2,0) -- (4,0);

\draw (2,0) -- (6,2);
\draw (2,0) -- (5,2);
\draw (2,0) -- (4,2);
\draw (2,0) -- (2,2);
\draw (2,0) -- (1,2);
\draw (2,0) -- (0,2);

\draw (4,0) -- (6,2);
\draw (4,0) -- (5,2);
\draw (4,0) -- (4,2);
\draw (4,0) -- (2,2);
\draw (4,0) -- (1,2);
\draw (4,0) -- (0,2);

\end{tikzpicture}
\caption{``Crown''-like graph}
\label{fig:crown}
\end{figure}
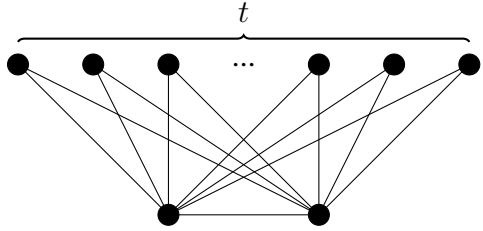

\subsection{Heavy-estimate data structure}

{
Next we describe a simple summary structure of the graph which we refer
to by $SE(q,l)$ here. An instance of
$SE(q,l)$ can be computed in one pass and can be used to decide whether an 
arbitrary edge of the graph is heavy or not in an approximate fashion.
It is formed as a collection of $l$ sets of edges $N_i$, chosen by
sampling as described in Algorithm~\ref{alg:buildse}.
These sampled edges are used to estimate whether a given edge $e$
meets the heaviness condition, using Algorithm~\ref{alg:heavy}
}

\begin{algorithm}[t]
For $i \in [l]$, in parallel.
\begin{itemize}
\item
  Sample each node independently with probability $q$ into $S_i$. 
\item
  $N_i \gets$ all edges incident on any node in $S_i$
\end{itemize}
\caption{Constructing the summary structure $SE(q,l)$}
\label{alg:buildse}
\end{algorithm}


\begin{algorithm}[t]
Given edge $e=(u,v)$ and summary structure $SE(q,l)$ (sampled
edge sets $N_i$):

For each $i\in [l]:$

\hspace{6mm}  $r_i(e) \gets |\{ w | (u,w) \in N_i \wedge (v,w) \in
N_i\}|$

\hspace{6mm}({\em count the number of triangles formed between $e$ and $N_i$})

Return $t'(e) = \operatorname{median}_i r_i(e)/q$ as estimate for $t(e)$

\caption{The heavy-estimate($e$) procedure}
\label{alg:heavy}
\end{algorithm}
 
{We prove the following lemma regarding Algorithm~\ref{alg:heavy}. }
\begin{lemma}
\label{lem:heavy-estimate}

Let $q \ge {16 \over \eps^2d}$ and $l=c\log n$ for some constant $c$. The procedure 
 heavy-estimate($e$) defined by Algorithm~\ref{alg:heavy} can be used to decide whether $t(e) \ge d$ or 
$t(e) < \frac{1}{4} d$ with high probability. Moreover for every edge with
$t(e) \ge d/4$, $t'(e)$ approximates $t(e)$
within a $(1+\eps)$ factor.
\end{lemma}

\begin{proof}
{ Fix an ordering on the triangles 
sharing $e$ and let $X_{i,j}$ be the random variable 
corresponding to the event of the $j$th triangle on $e$ being counted
in the $i$th edge set $N_i$.
That is, the $j$th triangle on $e$ is defined by the nodes $\{u, v,
w\}$, and $X_j$ is 1 if $w \in S_i$
 We have $\E[X_j]=q$ and
thus $\E[r_i(e)=\sum_{j=1}^{t(e)} X_j]=qt(e).$ Therefore, $\E[r(e)/q]=t(e)$. 
By a Chernoff bound,
$$\Pr[|t'(e)-t(e))|>\eps t(e))]\le e^{-\frac{\eps^2qt(e)}{4}}.$$
}

{Therefore, for $t(e) \ge d/4$ {and choosing $q \ge \frac{16}{\eps^2d}$},
each estimate $t'(e)$ is close to the correct value with constant
probability greater than $1/e$. 
Hence, taking the median of $O(\log n)$ instances gives us a value within
the desired bounds with probability $1 - O(1/n^2)$, via a standard Chernoff
bound argument. }

{On the other hand, for edges with $t(e) < \frac14d$, the Markov inequality
implies that $\Pr[r(e)/q > d] < 1/4$.
The probability that the median
of  $\Theta(\log n)$ repetitions of the estimator goes beyond $d$ is $O(\frac1{n^2})$.   }
\end{proof}

\subsection{Relative error approximation}
\label{sec:rel}

To obtain a relative error guarantee, we overlap the execution of
three algorithms in two passes, as detailed in
Algorithm~\ref{alg:rel}.
Note that to optimize the dependence on $\eps$, the parameters used to
invoke each of the Algorithms~\ref{alg:const} and \ref{alg:heavy}  are
carefully chosen. 

\begin{algorithm}[t]

Do the following tasks in parallel:

\begin{itemize}
\item
  Run Algorithm~\ref{alg:const} with $\eps=1/12$
to find $t'$ such that $1/4t \le t' \le t$.

\item
  Run Algorithm~\ref{alg:heavy} to compute $SE(q,l)$ for $q=\Theta({\eps^{-1.5}T^{-0.5}})$ 
  and $l=\Theta(\log n)$.

\item
Let $B = \{T, 2T, 4T, ..., 2^iT\}$ where $i$ is the smallest integer such that 
$2^iT \ge n^3$.

In the second pass, instantiate $|B|$ parallel instances of 
the algorithm PT$_b$ with parameters $T=b, J=24\sqrt{b/\eps}$ for all $b \in B$. 
Also initiate counters $\{h_b\}_{b \in B}$ with zero.

Upon receiving the edge $e \in E$, first compute $t'(e)$ using the
heavy-estimate procedure. For all $b \in B$, 
if $t'(e) \ge d=24\sqrt{{b}/{\eps}}$, 
we add $t'(e)$ to the global counter $h_{b}$, otherwise 
we feed $e$ to PT$_b$.

At the end of the pass, let $t_b$ be the 
output of PT$_b$. We output $h_j+t_j$ as the final estimate
for the number of triangles where $j=\max{\{b\in B \, | \, b \le t'\}}$.  
\end{itemize}
\caption{Relative Error Algorithm}
\label{alg:rel}
\end{algorithm}


\begin{theorem}
Algorithm~\ref{alg:rel} is a 2-pass randomized streaming algorithm that takes
$O(\frac{m}{\eps^{2.5}\sqrt{T}}\polylog(n))$ space in 
expectation and outputs a $(1+\eps)$
factor approximation of $t(G)$.
\end{theorem}

\begin{proof}

In the following, for the sake of simplicity in exposition, we
assume the randomized procedures used in the algorithm do not err. 
With appropriate choice of parameters the total error probability
can always be bounded by a constant smaller than $1/2$.  

As in Algorithm~\ref{alg:rel}, let  $j=\max{\{b\in B \, | \, b \le t'\}}$.
 Let $E_j$ be the set of edges that are identified as heavy by the algorithm for
parameter $j$. Also let $G_j$ be the graph $G$ after removing the edges $E_j$. 
Let $t(E_j)$ be the number of triangles in $G$ that share at least an edge with $E_j$.
Clearly $t(G)=t(E_j)+t(G_j)$. 

First, we prove that $h_j$ is indeed a $1+O(\eps)$ approximation of
$t(E_j)$. A source of error comes from the fact that we will be over-counting
triangles that have more than one heavy edges. We show that number of
such triangles is limited. To see this, 
observe that with the choice of parameters $q =
\Theta(\eps^{-1.5}T^{-0.5})$
and $d=3\sqrt{\frac{t}{2\eps}}$
in Lemma \ref{lem:heavy-estimate}, 
it follows that for each $e \in E_j$ we have
 $t(e) \ge 6\sqrt{\frac{j}{\eps}}$ and
 consequently $t(e) \ge 3\sqrt{{{t}\over {2\eps}}}.$ The latter follows from the fact that $t/8 \le j\le t$.
But, by Lemma \ref{lem:heavy}, the number of triangles that have 
two or three heavy edges that are shared by more than 
$3\sqrt{\frac{t}{2\eps}}$ is at most $\eps t$. 
Another source of error comes from estimation errors $|t(e)-t'(e)|$.
This is also negligible since, by Lemma \ref{lem:heavy-estimate}, for every identified
heavy edge $e$, we have $|t(e)-t'(e)|$ bounded by $\eps t(e)$.

 On the other hand, the maximum number of triangles on an edge
 for graph $G_j$ is at most $O(\sqrt{\frac{j}{\eps}})$. Hence
  by Lemma \ref{lem:PT}, $t_j$ estimates $t(G_j)$ within
 $\eps t$ additive error using $O(\frac{m}{\eps^{2.5}\sqrt{t}})$ space. 
 Consequently $h_j+t_j$ estimates
 $t(G)$ within $4\eps t$ additive error.
Rescaling $\eps$ gives the desired result.

 It remains to show that the expected space usage of the 
 algorithm is bounded as claimed. The $SE(q,l)$ summary
 takes $O(nq\frac{m}{n})=O(\frac {m}{\eps^{1.5}\sqrt{T}})$
 in expectation. The constant factor approximation takes $O(m/\sqrt{T})$ space. 
 The instance of the PT algorithm with parameter
 $b$ takes $O(\frac{m\sqrt{b/\eps}}{\eps^2b}+\frac{m}{\eps^2\sqrt{b}})$ in
 expectation which is bounded by $O(\frac{m}{\eps^{2.5}\sqrt{T}})$ as $T \le b$. 
\end{proof}

\subsection{One pass algorithm}

{It is natural to ask whether this algorithm can be reduced to a single
pass.
There are several obstacles to doing so.
Primarily, we need to determine for each edge whether or not it is heavy, and
handle it accordingly.
This is difficult if we have not yet seen the subsequent edges which
make it heavy.
We can adapt Algorithm~\ref{alg:const} to one pass to obtain a constant
factor approximation, under the assumption of a randomly ordered
stream. }

\begin{corollary} 
\label{cor:sqrtt}
Assuming the data arrives in random order, 
there is a one-pass randomized streaming algorithm that 
returns a $1/3+\eps$ factor approximation of 
$t(G)$ that 
uses $O( {\frac {m}{\eps^{4.5}\sqrt{T}}})$ space. 
\end{corollary}

\begin{proof}
Under random order, we can combine the first and second passes of
Algorithm~\ref{alg:const}. 
That is, we sample edges with probability $p$, and look for triangles
observed based on the stream and the sample. 
We count all triangles formed as $r$: either those with all three edges
sampled, or those with two edges sampled and the third observed
subsequently in the stream.  
The estimator is now $\frac{r}{p^2}$, 
since the probability of counting any triangle
is $p^3$ (for all three edges sampled) plus $p^2(1-p)$ (for the first
two edges in the stream sampled, and the third unsampled). 
The same analysis as for Theorem~\ref{thm:sqrtt} then follows: 
we partition the  edges in to light and heavy sets, and bound the
probability of sampling a subset of triangles. 
A triangle with two light edges is counted if both light edges are
sampled, and the heavy edge arrives last. 
This happens with probability $p^2/3$. 
We can nevertheless argue that we are unlikely to undercount such
triangles, following the same Chebyshev analysis as above. 
This allows us to conclude that the estimator is good. 
\end{proof}

We emphasize random order is critical to make this algorithm work in one
pass: an adversarial order could arrange the heavy edges to always
come last (increasing the probability of counting a triangle under
this analysis) or always first (giving zero probability of counting a
triangle under this analysis).

\subsection{Non-simple graphs}

All our algorithms can be modified to work when the graphs are not
simple.
That is, we may see the same edge multiple times in the graph stream,
but are only interested in counting each unique triangle once.
We need two tools to accomplish this: (1) hash functions which map nodes
or edges to real numbers in the range $[0\ldots 1]$ and can be treated
as random (2) count-distinct algorithms which can approximate the number of
unique items (tuples of nodes) that are passed to them, up to a
$(1+\eps)$ factor.

The transformation of the algorithms is to replace sampling with
hashing and testing if the hash value is less than the threshold $p$.
This has the effect of sampling each unique edge (or node) with
probability $p$.
We replace counting triangles with a count-distinct of the triangles.

For example, Algorithm~\ref{alg:const} uses hashing to determine which
(distinct) edges to sample in pass 1, then approximately counts the
set of distinct triangles in pass 2.
The one-pass algorithm of Lemma~\ref{lem:PT} can correspondingly be
modified, as can Algorithm~\ref{alg:heavy}.
The main change needed for 
Algorithm~\ref{alg:rel} is that we should extract the {\em set} of triangles
counted in each invocation of Algorithm~\ref{alg:heavy}, and pass
these to an instance of a count-distinct algorithm $h_b$.
Consequently, our results also apply to the case of repeated edges.
The space cost increases due to replacing counters with approximate
counters.
In each algorithm, the number of instances of count-distinct
algorithms is small.
Hence, these modifications 
increase the space cost by an additional $\tilde{O}(1/\eps^2)$, which
does not change the asymptotic bounds. 

\section{Lower bounds}

We now show lower bounds for the problem $\dist(T)$, to distinguish 
between the case $t=0$ and $t\ge T$. 
Our first result builds upon a lower bound from prior work, and amplifies the
hardness. 
We formally state the previous result:

\begin{lemma} \cite{Braverman:Ostrovsky:Vilenchik:13}
\label{lem:bovlb2}
Every constant pass streaming algorithm for $\dist(T)$ 
requires $\Omega(\frac mT)$ space.
\end{lemma}

\begin{theorem} 
\label{thm:lb1}
Any constant pass streaming algorithm for $\dist(T)$ 
requires $\Omega(\frac m{T^{2/3}})$ space.
\end{theorem}

\begin{proof}  
Given a graph $G=(V,E)$ with $m$ edges
 we can create a graph $G'=(V',E')$ with $mT^2$ edges and $t(G') =
 T^3t(G)$. 
We do so
 by replacing each vertex $v \in V$ with $T$ vertices $\{v_1,\ldots,v_T\}$ and replacing
 the edge $(u,v) \in E$ with the edge set $\{u_1,\ldots,u_T\}\times \{v_1,\ldots,v_T\}$.
 Clearly any triangle in $G$ will be replaced by $T^3$ triangles
 in $G'$ and every triangle in $G'$ corresponds to a triangle in $G$.
Moreover this reduction can be peformed in a streaming fashion using $O(1)$ space.
Therefore a streaming algorithm for $\dist(T)$ 
using $o(\frac m{T^{2/3}})$ (applied to $G'$) 
would imply an $o(m)$ streaming algorithm for $\dist(1)$. 
But from Lemma \ref{lem:bovlb2}, we have that 
$\dist(1)$ requires $\Omega(m)$ space for constant pass algorithms. 
This is a contradiction and as a result our claim is proved.
\end{proof}

Our next lower bound more directly shows the hardness by a reduction to the
hard communication problem of $\disj_p^{r}$. 

\begin{theorem} 
\label{thm:lb2}
For any $\rho > 0$ and $T \le n^2$, there is no constant pass streaming algorithm for $\dist(T)$ 
that takes $O(\frac m{T^{1/2+\rho}})$ space.
\end{theorem}

\begin{proof} 
We show that there are families of graphs
with $\Theta(n\sqrt{T})$ edges and $T$ triangles such that distinguishing them
from triangle-free graphs in a constant number of passes requires $\Omega(n)$ space.
This is enough to prove our theorem. 

We use a reduction from the standard set intersection problem, here denoted by $\disj^{n/2}_n$.
Given $y \in \{0,1\}^n$, Bob constructs a bipartite graph $G=(A \cup B,E)$ where
$A = \{a_1,\ldots,a_n\}$ and $B=\{b_1,\ldots,b_{\sqrt{T}}\}$. 
He connects $a_i$ to all vertices in $B$ iff $y[i]=1$. 
On the other hand, Alice adds vertices $C=\{c_1,\ldots,c_{\sqrt{T}}\}$
to $G$. 
She adds the edge set $C \times B$. 
Also for each $i \in [\sqrt{T}]$ and $j \in [n]$, she adds the edge
 $(c_i,a_j)$ iff $x[j]=1$. 
We observe that if $x$ and $y$ (uniquely) intersect there will be
precisely $T$ triangles passing through each vertex of $C$. 
Since there is no edge between the vertices in $C$, in total we will have
 $T$ triangles. 
On the other hand, if $x$ and $y$ represent disjoint sets,
 there will be no triangles in $G$. 
In both cases, the number of edges is between $2n\sqrt{T}$ and
$3n\sqrt{T}$, over $O(n)$ vertices (using the bound $T^2 \leq n$). 
 Considering the lower bound for the $\disj_p^{r}$
 (Section~\ref{sec:prelims}), our claim is proved following a standard
 argument: a space efficient streaming algorithm would imply an efficient
 communication protocol whose messages are the memory state of the
 algorithm. 
\end{proof}

\paragraph{Acknowledgments}  
We thank Andrew McGregor, Srikanta Tirthapura and Vladimir Braverman
for several helpful conversations. 
We also thank Andrew McGregor and Sofya Vorotnikova for alerting us to the error in the first version of this work. 

\bibliographystyle{amsalpha}
\bibliography{references}{}

\end{document}